\documentclass{ifacconf}
\usepackage{cite}
\usepackage{amsmath,amssymb,amsfonts}
\usepackage{graphicx}
\usepackage{natbib}        
\usepackage{mathtools}
\usepackage{bm}
\usepackage{mathrsfs}
\usepackage{kantlipsum}
\usepackage[shortlabels]{enumitem}

\begin{document}
\begin{frontmatter}

\title{Geometric Fault Identification via Mirror Descent Learning\thanksref{footnoteinfo}} 

\thanks[footnoteinfo]{This research is funded in part by the Technology Innovation Institute and the Defense Advanced Research Projects Agency (Learning Introspective Control).}

\author[First]{Mahdi Taheri} 
\author[First]{Haeyoon Han} 
\author[First]{Soon-Jo Chung}
\author[First]{Fred Y. Hadaegh}

\address[First]{Division of Engineering and Applied Science, California Institute of Technology, Pasadena, CA 91125 USA (e-mail: mtaheri, hhan3, sjchung, hadaegh\}@caltech.edu).}

\begin{abstract}                
This paper develops a fault detection and identification (FDI) method for nonlinear control-affine systems under simultaneous actuator and sensor faults. We adopt a geometric approach to study the isolability of faults in the sense of the principal angles between subspaces corresponding to each actuator and sensor fault. As for the fault identification, a hybrid estimator that consists of a Luenberger-like observer with contraction guarantees is developed. Moreover, neural networks are embedded in the mentioned observer to estimate actuator and sensor faults. Considering that the training dataset for neural networks cannot be representative of every fault scenario, the last layer of each network is adapted using mirror descent-based laws. The mirror descent-based adaptive laws impose isolability conditions for fault channels and do not assume a quadratic parameter estimation space to consider the geometry of the fault subspaces. A Lyapunov-based analysis establishes that the state and parameter estimation errors are uniformly ultimately bounded. The effectiveness of our proposed FDI method is illustrated on the $3$-axis attitude control system of a spacecraft.\end{abstract}


\end{frontmatter}

\section{Introduction}
Fault detection and identification (FDI) is a fundamental capability for autonomous systems, which ensures reliable performance in the presence of unexpected component failures. A desirable FDI system should detect deviations from nominal behavior and identify faults by isolating the fault type and estimating its severity. In particular, on-board FDI is critical for actuator and sensor faults, which have a cascading impact on a fast time scale and may trigger catastrophic failures if left unmitigated (\cite{10.1145/3146389}). For this reason, on-board FDI for actuators and sensors has become a vital topic across various domains, including autonomous vehicles (\cite{mohamed2018literature}), aerospace systems (\cite{10.1007/978-981-97-7094-6_3}), and industrial robotics (\cite{khan2025fault}). 

Existing FDI methods for autonomous systems are broadly divided into three categories: model-based, data-driven, and hybrid approaches that integrate prior knowledge of the system model and adaptive capabilities learned from data. Model-based FDI relies on mathematical physics-based models that capture system dynamics and estimate faults by generating residuals between the predicted nominal and measured system output (\cite{depersis2001geometric}). However, uncertainty in system dynamics (e.g., disturbances, unmodeled dynamics) makes model-based FDI challenging. To address this issue and improve the accuracy of the FDI, robust fault observers (\cite{melendez-useros_novel_2025, moradmand2020a}), information-gathering tree search with marginalized filtering (\cite{ragan_online_2024}), and a bank of estimators (\cite{995036}) have been developed.

Although providing structural ways for fault diagnosis, model-based methods require accurate analytical models, which may be unavailable for complex, nonlinear systems. In such settings, data-driven approaches offer alternatives by leveraging past faulty and healthy historical datasets to learn fault signatures. Recent works on data-driven FDI have demonstrated that neural networks trained with input-output data of autonomous systems can perform fault detection (\cite{chen_data-driven_2022}), isolation (\cite{kumar_recurrent_2023}), and identification (\cite{elhoseny_deep_2024}).

\begin{figure*}[thpb]
\begin{center}
\centering    \includegraphics[width=0.89\textwidth, clip,keepaspectratio]{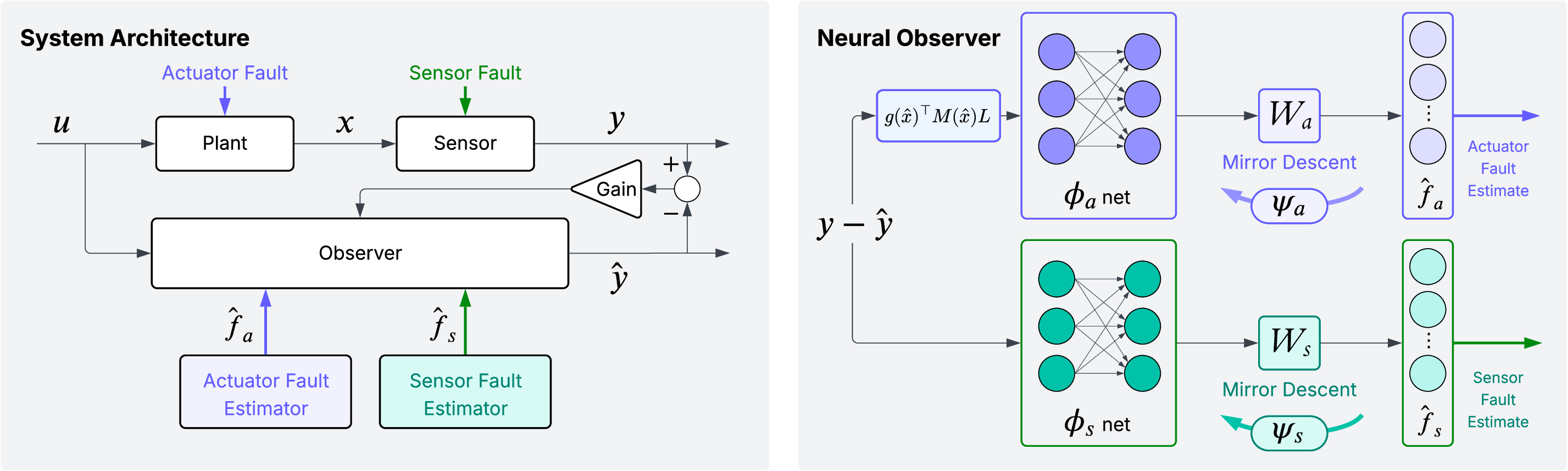}
\caption{Overall architecture of the proposed fault detection and identification scheme. A contraction-based neural observer with deep neural networks, trained offline to estimate actuator and sensor faults, and their last layers are adapted online via mirror descent, so that the learned fault estimates geometrically isolable fault directions.}\label{fig:1}
\label{fig:diagram}
\end{center}
\vspace{-2mm}
\end{figure*}

This paper presents a hybrid approach for FDI designed for autonomous systems. Hybrid approaches retain the interpretability and structure of model-based reasoning, while adapting the uncertain system parameters through data-driven inference. Prior works have explored hybrid FDI methods:~\cite{alessandri2003fault} used a bank of estimators in which each estimator is a neural network approximating the optimal model for describing the system and failures. Similar multi-model based methods were further developed in~\cite{995036, sobhani-tehrani_hybrid_2014}. Another hybrid approach was introduced by \cite{talebi_recurrent_2009}, which utilizes neural network-based residual estimation to prevent the cascading failures. A more advanced learning-based fault estimation was proposed by~\cite{oconnell_learning-based_2024}, which estimates the actuator fault from the system's dynamic response, where the influences of actuator fault and unmodeled dynamics are coupled, supporting a fault-tolerant controller. When the system model's parameters are unknown, hybrid approaches may employ data-driven system identification, and a residual generator can then be designed based on the updated model structure, as shown in~\cite{ran_hybrid_2023}.

The fault observer developed in this paper uses composite adaptation: a deep neural network learns fault dynamics representation from fault data, and the last layer allows the fault parameter estimator to be updated via mirror descent. While fault observers in conventional data-driven approaches commonly adopt Euclidean distances to represent estimation error and update parameters via gradient descent, mirror descent employed in this paper uses Bregman divergence to quantify the closeness between parameter estimates. This enables the adaptation law to consider and exploit the geometry of the fault signature subspaces. The advantage of using mirror descent has been shown in control of autonomous systems with regularization: \cite{fradkov_lyapunov-bregman_2022} designed nonlinear adaptive control laws with Bregman divergence to preserve Riemannian geometry, and \cite{9812089} achieved efficient online learning for Model Predictive Control using dynamic mirror descent.

The contributions of this paper are summarized as follows:
  \begin{enumerate}[wide, labelwidth=!, labelindent=0pt]
  \item  We develop a geometric framework in which actuator and sensor faults are represented as fault signature subspaces in the output space of a nonlinear system, and fault isolability is described in terms of principal angles between these subspaces. This framework accounts for simultaneous actuator and sensor fault scenarios.
\item  A contraction-based Luenberger-like observer augmented with neural network estimators for actuator and sensor faults is designed, where only the last layer is adapted online using mirror descent laws that are derived from the Bregman divergence with an elastic net potential function. This design allows the parameter updates to exploit the geometry of the fault signature subspaces and bias adaptation along geometrically isolable directions without assuming a purely quadratic parameter space.
\item  We derive uniform ultimate boundedness (UUB) stability of the designed adaptive observer. Accounting for Bregman divergence used for the parameter update, the stability analysis demonstrates UUB with exponential convergence for the state estimation.
  \end{enumerate}

\section{Problem Formulation}
\subsection{Notation and Preliminaries}
Let $\|\cdot\|$ and $\langle \cdot,\cdot\rangle$ denote the Euclidean norm and inner product. For a matrix $A$, $\mathrm{Im}\,A$ denotes its column space and $A^\top$ its transpose. Let $x\in\mathbb{R}^n$ denote a state variable. For a smooth function $h:\mathbb{R}^n~\to~\mathbb{R}^p$ and the vector field $v$, the Lie derivative of $h$ along $v$ is $L_v h(x) := \frac{\partial h}{\partial x}(x)\,v(x) \in \mathbb{R}^p$. Higher-order Lie derivatives are defined recursively as $L_v^0 h := h$ and $L_v^{k+1} h := L_v(L_v^k h)$. For a smooth map $F$, the Jacobian with respect to $x$ is written as $\nabla_x F(x)$.

\subsection{System Model}
Consider a nonlinear control-affine system expressed by
\begin{align}\label{e:sys}
		\dot x &= f(x) + \sum_{k=1}^m g_k(x)u_k+ \sum_{i=1}^{q_\mathrm{a}} g^\mathrm{f}_i(x) f_{\mathrm{a},i}(t) + d(t), \nonumber \\
		y &= h(x) + \sum_{j=1}^{q_\mathrm{s}} e_j f_{\mathrm{s},j}(t),
\end{align}
where $x\in M\subset\mathbb{R}^n$ is the state on a smooth manifold $M$, $u =[u_1,\dots,u_m]^\top\in\mathbb{R}^m$ is the control input associated with $m$ actuators, $y\in\mathbb{R}^p$ is the measured output, $d(t)\in\mathbb{R}^n$ is a bounded disturbance (i.e., $\|d(t)\|<\bar{d}$), and $f_{\mathrm{a},i}(t)\in\mathbb{R}$ and $f_{\mathrm{s},j}(t)\in~\mathbb{R}$ are actuator and sensor fault channels (i.e., signals), respectively, for $i=1,\dots,q_a$ and $j=1,\dots,q_s$. It should be noted that we do not impose a specific functional form on the fault signals $f_{\mathrm{a},i}(t)$ and $f_{\mathrm{s},j}(t)$. Moreover, $g_i^\mathrm{f}$ and $e_j$ denote the actuator and sensor fault signatures, respectively. The index $i$ enumerates fault channels, not physical actuators. The assignment from physical components to fault channels is encoded by the choice of $g_i^\mathrm{f}$ and $e_j$. For example, if the first actuator undergoes a loss of effectiveness fault, this fault is represented by selecting some channel $i$ and setting $g_i^\mathrm{f}(x)=g_1(x)$, with $f_{\mathrm{a},i}(t)$ acting as the corresponding fault signal. Sensor faults are encoded analogously, with the nonzero rows of $e_j$ indicating which physical sensors are affected by the $j$-th channel. We assume that $f,g_j,g_i^\mathrm{f}$, and $h$ are smooth maps of appropriate dimensions.

\subsection{Problem Statement}
We study and address three tightly coupled problems. First, we investigate conditions under which each actuator and sensor fault channel $f_{\mathrm{a},i}(t)$ and $f_{\mathrm{s},j}(t)$ is structurally isolable in the sense that its effect on the measured outputs can be distinguished from the nominal dynamics and from other concurrent faults. This is formulated by adopting a differential-geometric approach, where fault effects are represented as subspaces of the linearized output differential map associated with~\eqref{e:sys} (made precise in Section~\ref{s:subspaces}). The isolability between faults is discussed by means of principal angles between the mentioned subspaces. Second, in addition to this geometric characterization, we seek to design a nonlinear observer that incorporates neural network-based fault estimators. In particular, the last layers of the neural network estimators are adapted in real-time via mirror descent laws derived from a suitable Bregman divergence. Third, we aim to derive Lyapunov–type guarantees to show UUB of both state and parameter errors, and to make explicit how the geometric principal angle conditions enter these stability and isolation properties.

\section{Fault Signature Subspaces and Principal Angle Isolability}\label{s:subspaces}
We need to derive a differential-geometric representation of fault signatures in an output differential map. This map is used to separate nominal and fault-induced directions and study the conditions under which faults can be isolated from one another.

\subsection{Relative Degree and Output Differential Maps}
Let $d(t)=0$ in \eqref{e:sys}. For each output $y_\ell = h_\ell(x)$, where $y_\ell$ denotes the $\ell$-th row of $y$, we define a relative degree $r_\ell \ge 1$ as the smallest integer such that a component of the control input $u$ appears explicitly in $\tfrac{d^{r_\ell}}{dt^{r_\ell}} h_\ell(x)$, i.e., the smallest $r_\ell$ such that
\begin{equation*}
	L_{g_k} L_f^{r_\ell-1} h_\ell(x) \neq 0,
\end{equation*}
for at least one $k\in\{1,\dots,m\}$. We assume all relative degrees are finite and define
\begin{equation}\label{e:R}
	R \;\ge\; \max_{\ell=1,\dots,p} r_\ell.
\end{equation}

The output differential map of order $R$ is the stacked vector
\begin{equation}\label{e:diffMap}
	y^{[0:R]}(x,f_\mathrm{a},f_\mathrm{s}) := ( y,\ \dot y,\ \dots,\ y^{(R)})\in \mathbb{R}^{p(R+1)},
\end{equation}
where $f_\mathrm{a}=(f_{\mathrm{a},1},\dots,f_{\mathrm{a},q_\mathrm{a}})^\top$ and $f_\mathrm{s}=(f_{\mathrm{s},1},\dots,f_{\mathrm{s},q_\mathrm{s}})^\top$. 

\begin{rem}\label{rem:reldeg}
We adopt a per-output relative degree rather than the vector relative degree to identify the smallest order of differentiation of each output at which control and fault-induced directions appear. The Lie derivative condition $L_{g_k}L_f^{r_\ell-1}h_\ell(x)\neq 0$ is required to hold on an open neighborhood of the operating point $x^*$ used in Section~\ref{s:subspaces}, which is generic for systems analytic in $x$. The integer $R\ge \max_\ell r_\ell$ in~\eqref{e:R} is used solely to obtain a single, uniform order for the stacked map~\eqref{e:diffMap}. 
\end{rem}

Each derivative $y^{(k)}(t)$ can be written as a finite combination of Lie derivatives of $h$ along $f$, the control vector fields $g_k$, and, when faults are present, the fault vector fields $g^f_i$.

\begin{rem}\label{rem:diffOutputMap}
	It should be noted that we utilize the differential map \eqref{e:diffMap} as a local, instantaneous representation of how the state and each fault channel affect the outputs. However, the observer and neural network-based fault estimators do not implement explicit time differentiation of $y$. The map \eqref{e:diffMap} is only used for theoretical analysis of fault signatures to formally define fault subspaces and their isolability conditions.
\end{rem}
\subsection{Output Differential Map and Fault Subspaces}
Throughout this section, $x^*$ denotes an instantaneous evaluation $x^*(t)$ along a fault-free trajectory of~\eqref{e:sys} at a fixed time $t$. We now study the first-order variation of $y^{[0:R]}(x,f_\mathrm{a},f_\mathrm{s})$ around a nominal fault-free trajectory $\big(x^*(t),u^*(t)\big)$. Consider small perturbations $(\delta x,\delta f_a,\delta f_s)$ around $(x^*,0,0)$. 
\begin{align}\label{eq:delta_diffMap}
	\delta y^{[0:R]}\!=\! C_R(x^*\!)\,\delta x \!+\!\! \sum_{i=1}^{q_\mathrm{a}}\! F_{\mathrm{a},i}^{[0:R]}(x^*\!)\,\delta f_{\mathrm{a},i} 
	\!+\!\! \sum_{j=1}^{q_s}\! E_{s,j}^{[0:R]}(x^*\!)\,\delta f_{s,j},
\end{align}
where $C_R(x^*)\in\mathbb{R}^{p(R+1)\times n}$ is the Jacobian of the fault-free output differential map \eqref{e:diffMap} with respect to $x$, $F_{\mathrm{a},i}^{[0:R]}(x^*)\in~\mathbb{R}^{p(R+1)\times 1}$ is the partial derivatives of $y^{[0:R]}$ with respect to the actuator fault signal at $f_{\mathrm{a},i}=0$, and $E_{\mathrm{s},j}^{[0:R]}(x^*)\in~\mathbb{R}^{p(R+1)\times 1}$ is the partial derivative with respect to $f_{\mathrm{s},j}$ and at $f_{\mathrm{s},j}=0$. We now introduce the subspaces spanned by the columns of these matrices in the output differential space.

\begin{defn}[Fault Signature Subspaces]\label{def:subspaces}
	At $x^*\in~\mathbb{R}^n$, we define the subspaces
		\begin{align*}
			\Delta_\mathrm{d}^{[0:R]}(x^*)&:= \mathrm{Im}\,C_R(x^*), \ \ \mathcal{F}_{\mathrm{a},i}^{[0:R]}(x^*)	:= \mathrm{Im}\,F_{\mathrm{a},i}^{[0:R]}(x^*),\\
						\mathcal{E}_{\mathrm{s},j}^{[0:R]}(x^*)	&:= \mathrm{Im}\,E_{\mathrm{s},j}^{[0:R]}(x^*), \quad i,j=1,\dots,q_\mathrm{s}.
		\end{align*}
\end{defn}

By construction, $\Delta_\mathrm{d}^{[0:R]}$ contains only the directions generated by the fault-free dynamics, while actuator and sensor faults reside in $\mathcal{F}_{\mathrm{a},i}^{[0:R]}$ and $\mathcal{E}_{\mathrm{s},j}^{[0:R]}$, respectively. Consequently, for each fault channel, one can define residue subspaces as the span of all other signatures.

\begin{defn}[Residue Subspaces]\label{def:residue}
	For an actuator fault channel $i\in\{1,\dots,q_\mathrm{a}\}$, we define
	\begin{equation*}
		\Pi_i^{[0:R]}(x^*):= \mathcal{F}_{\mathrm{a},-i}^{[0:R]}(x^*) \oplus  \bar{\mathcal{E}}_{\mathrm{s}}^{[0:R]}(x^*),
	\end{equation*}
	where $\oplus$ denotes the algebraic sum of subspaces, $\mathcal{F}_{\mathrm{a},-i}^{[0:R]}(x^*)= \bigoplus_{\substack{i_\mathrm{a}=1 \\ i_\mathrm{a}\neq i}}^{q_{\mathrm{a}}}\mathcal{F}_{\mathrm{a},i_\mathrm{a}}^{[0:R]}(x^*)$, and $\bar{\mathcal{E}}_{\mathrm{s}}^{[0:R]}(x^*) = \bigoplus_{j=1}^{q_\mathrm{s}} \mathcal{E}_{\mathrm{s},j}^{[0:R]}(x^*)$. Moreover, for a sensor fault channel $j\in\{1,\dots,q_\mathrm{s}\}$, we define
	\begin{equation}
		\Sigma_j^{[0:R]}(x^*):= \bar{\mathcal{F}}_{\mathrm{a}}^{[0:R]}(x^*) \oplus \mathcal{E}_{\mathrm{s},-j}^{[0:R]}(x^*),
	\end{equation}
	where $\bar{\mathcal{F}}_{\mathrm{a}}^{[0:R]}(x^*)= \bigoplus_{i=1}^{q_\mathrm{a}} \mathcal{F}_{\mathrm{a},i}^{[0:R]}(x^*)$ and $\mathcal{E}_{\mathrm{s},-j}^{[0:R]}(x^*)= \bigoplus_{\substack{j_\mathrm{s}=1 \\ j_\mathrm{s}\neq j}}^{q_{\mathrm{s}}}\mathcal{E}_{\mathrm{s},j_\mathrm{s}}^{[0:R]}(x^*)$.
\end{defn}

By definition, $\Pi_i^{[0:R]}(x^*)$ excludes the target actuator subspace $\mathcal{F}_{\mathrm{a},i}^{[0:R]}(x^*)$, and similarly, $\Sigma_j^{[0:R]}(x^*)$ excludes the target sensor subspace $\mathcal{E}_{\mathrm{s},j}^{[0:R]}(x^*)$. These objects will be used below to formulate necessary and sufficient fault isolability conditions in terms of principal angles and intersections between subspaces.

\subsection{Principal Angle Isolability in Output Map Space}
We now utilize the fault subspaces introduced in Definition~\ref{def:residue} to derive necessary and sufficient geometric conditions on the isolability of fault channels based on principal angles between their subspaces. Let $U,V\subset\mathbb{R}^N$ be two subspaces with orthonormal bases $Q_\mathrm{U}\in\mathbb{R}^{N\times r_U}$ and $Q_\mathrm{V}\in\mathbb{R}^{N\times r_V}$, respectively. The principal angles $\theta_1~\le~ \dots\le \theta_{r_\mathrm{min}}$, where $r_\mathrm{min}:=\min\{r_U,r_V\}$, are defined via the singular values of $Q_\mathrm{U}^\top Q_\mathrm{V}$, where $Q_\mathrm{U}^\top Q_\mathrm{V} = U \mathrm{S} V^\top$, $\mathrm{S} = \mathrm{diag}(\sigma_1,\dots,\sigma_{r_\mathrm{min}})$, with $\sigma_k = \cos\theta_k\in[0,1]$ and $\theta_k\in[0,\tfrac{\pi}{2}]$, for $k=1,\dots,r_\mathrm{min}$. We denote $\theta_{\min}(U,V):=\theta_1$, where $\cos\theta_{\min}(U,V)=\sigma_1$. It can be simply seen that $U\cap V = \{0\}$, if and only if $\theta_{\min}(U,V) > 0$.

In order to study the principal angle between fault subspaces, we need to first define fault isolability among various fault channels.
\begin{defn}[Fault Isolability]\label{def:isolability}
	Actuator channel $i$ is structurally isolable at $x^*$ if there exists a matrix	$H_{\mathrm{a},i}^{[0:R]} \in \mathbb{R}^{\ell_i\times p(R+1)}$ such that:
	\begin{enumerate}[(a)]
		\item For every $v\in \Pi_i^{[0:R]}(x^*)$, $H_{\mathrm{a},i}^{[0:R]} v = 0$ .
		\item The restriction of $H_{\mathrm{a},i}^{[0:R]}$ to $\mathcal{F}_{\mathrm{a},i}^{[0:R]}(x^*)$ is injective, i.e.,
		\begin{equation*}
			\mathrm{rank}\big(H_{\mathrm{a},i}^{[0:R]} F_{\mathrm{a},i}^{[0:R]}(x^*)\big)= \dim \mathcal{F}_{\mathrm{a},i}^{[0:R]}(x^*).
		\end{equation*}
	\end{enumerate}
	In a similar manner, the $j$-th sensor channel is isolable if there exists $H_{\mathrm{s},j}^{[0:R]}$ such that 
	$H_{\mathrm{s},j}^{[0:R]} v=0$ for all $v\in\Sigma_j^{[0:R]}(x^*)$ and $\mathrm{rank}\big(H_{\mathrm{s},j}^{[0:R]} E_{\mathrm{s},j}^{[0:R]}(x^*)\big) = \dim \mathcal{E}_{\mathrm{s},j}^{[0:R]}(x^*)$.
\end{defn}

We now state the main result of this section.

\begin{thm}\label{th:isolability}
	Let $x=x^*$ and consider an integer $R$ as in \eqref{e:diffMap}. The following can be stated regarding the isolability of actuator and sensor faults:
	\begin{enumerate}[(a)]
		\item Actuator channel $i$ is isolable in the sense of Definition~\ref{def:isolability} if and only if
		\begin{equation}\label{e:th_act_isolability}
			\mathcal{F}_{\mathrm{a},i}^{[0:R]}(x^*)\ \cap\ \Pi_i^{[0:R]}(x^*) = \{0\}.
		\end{equation}
		\item Sensor channel $j$ is isolable if and only if
		\begin{equation}\label{e:th_sen_intersection}
			\mathcal{E}_{\mathrm{s},j}^{[0:R]}(x^*)\ \cap\ \Sigma_j^{[0:R]}(x^*) = \{0\}.
		\end{equation}
	\end{enumerate}
	Equivalently, isolability holds if and only if
	\begin{equation}\label{e:th_angle_isolability}
		\theta_{\min}\big(\mathcal{F}_{\mathrm{a},i}^{[0:R]},\Pi_i^{[0:R]}\big) > 0,\quad
		\theta_{\min}\big(\mathcal{E}_{\mathrm{s},j}^{[0:R]},\Sigma_j^{[0:R]}\big) > 0.
	\end{equation}
\end{thm}
\begin{pf}
	We prove the actuator fault case here. The sensor fault isolability follows along similar lines.

	\textbf{Sufficiency:} Suppose~\eqref{e:th_act_isolability} holds. There exists a direct sum decomposition as $\mathbb{R}^{p(R+1)}=\Pi_i^{[0:R]}(x^*) \oplus W_i^{[0:R]}$, where $W_i^{[0:R]}$ is a complementary subspace. Consequently, every vector $v\in\mathbb{R}^{p(R+1)}$ can be written uniquely as $v = v_\Pi + v_W$, where $v_\Pi\in\Pi_i^{[0:R]}(x^*)$ and $v_W\in W_i^{[0:R]}$.
	
	Let us define the linear projection $P_{W_i^{[0:R]}} : \mathbb{R}^{p(R+1)} \to \mathbb{R}^{p(R+1)}$, where	$P_{W_i^{[0:R]}} v = v_W$. By construction, we have $P_{W_i^{[0:R]}} v=0$ for all $v\in\Pi_i^{[0:R]}(x^*)$. Moreover, $P_{W_i^{[0:R]}}$ is injective on	$\mathcal{F}_{a,i}^{[0:R]}(x^*)$ since for $v\in\mathcal{F}_{\mathrm{a},i}^{[0:R]}(x^*)$ that satisfies $P_{W_i^{[0:R]}} v = 0$, one has $v=v_\Pi \in \Pi_i^{[0:R]}(x^*)$. Hence, $v\in\mathcal{F}_{\mathrm{a},i}^{[0:R]}(x^*)\cap\Pi_i^{[0:R]}(x^*)= \{0\}$. 
		
	Let $H_{\mathrm{a},i}^{[0:R]} = P_{W_i^{[0:R]}}$. Hence, $H_{\mathrm{a},i}^{[0:R]} v = 0$ for all
	$v\in\Pi_i^{[0:R]}(x^*)$. Thus, conditions in Definition~\ref{def:isolability} are satisfied, and the $i$-th actuator channel is isolable.
	
	\textbf{Necessity:} Assume that there exists $H_{\mathrm{a},i}^{[0:R]}$ satisfying Definition~\ref{def:isolability}. Moreover, suppose $\mathcal{F}_{\mathrm{a},i}^{[0:R]}(x^*) \cap \Pi_i^{[0:R]}(x^*) \neq~\{0\}$. Hence, there exists $0\neq v\in \mathcal{F}_{\mathrm{a},i}^{[0:R]}(x^*)\cap\Pi_i^{[0:R]}(x^*)$. Since $v\in\Pi_i^{[0:R]}(x^*)$, condition in Definition~\ref{def:isolability} imply that $H_{\mathrm{a},i}^{[0:R]} v = 0$. This contradicts the assumption since $v\in\mathcal{F}_{\mathrm{a},i}^{[0:R]}(x^*)$ and $H_{\mathrm{a},i}^{[0:R]}$ is injective on $\mathcal{F}_{\mathrm{a},i}^{[0:R]}(x^*)$. Hence, \eqref{e:th_act_isolability} follows.
	
	Moreover, since $U\cap V = \{0\}$ is equivalent to having $\theta_{\min}(U,V)>0$, \eqref{e:th_angle_isolability} follows immediately. \qed
\end{pf}

\begin{rem}
		The geometric isolability analysis in this subsection is based on the first-order differential output map and therefore on the Jacobians of $f$, $g_k$, and $h$ evaluated along a fault-free trajectory $x^*(t)$. This is consistent with standard treatments of nonlinear observability and decoupling problems (\cite{isidori1985nonlinear,1101601}), where structural properties are characterized locally and along trajectories. In particular, $x^*(\cdot)$ denotes any solution of	the fault-free dynamics under the given input $u^*(\cdot)$. The Jacobian-based linearization at $(x^*(t),u^*(t))$ is therefore used as an instantaneous approximation of how faults and nominal dynamics affect the output differential map at that operating point. The subsequent isolability	conditions are thus local structural conditions and do not require selecting a single global reference trajectory.
	\end{rem}
\subsection{Non-Isolability of Faults}
In Theorem~\ref{th:isolability}, conditions under which isolability conditions for actuator and sensor fault were studied. Considering that under certain conditions, isolation cannot be achieved, we investigate these in the sense of the principal angles.

\begin{thm}[Non-Isolability of Faults]\label{th:non-isolability}
	Consider $x=x^*$ and let $R$ be chosen as in \eqref{e:diffMap}.	If for the $i$-th actuator channel there exists a time $t_0$ such that
	\begin{equation}
		\theta_{\min}\big(\mathcal{F}_{\mathrm{a},i}^{[0:R]}(x^*(t_0)), \Pi_i^{[0:R]}(x^*(t_0))\big) =0,
	\end{equation}
	the actuator channel $i$ is not structurally isolable at time $t_0$. An analogous statement holds for sensor channels with $\mathcal{E}_{\mathrm{s},j}^{[0:R]}$ and $\Sigma_j^{[0:R]}$.
\end{thm}

\begin{pf}
	If $\theta_{\min}(U,V)=0$ for subspaces $U,V$, $U\cap V\neq\{0\}$. Let $U=\mathcal{F}_{\mathrm{a},i}^{[0:R]}(x^*(t_0))$ and $V=\Pi_i^{[0:R]}(x^*(t_0))$. Hence, there exists a non-zero $v\in\mathcal{F}_{\mathrm{a},i}^{[0:R]}\cap\Pi_i^{[0:R]}$. In light of Theorem~\ref{th:isolability}, the $i$-th channel fault cannot be isolated. \qed
\end{pf}

Theorem~\ref{th:non-isolability} shows that if the $i$-th fault signature becomes linearly dependent with its residue subspace in the output differential map \eqref{e:diffMap}, there exists no matrix $H_{\mathrm{a},i}^{[0:R]}$ such that $H_{\mathrm{a},i}^{[0:R]} y^{[0:R]}$ satisfies the isolability conditions of Definition~\ref{def:isolability} at that point. In other words, within the considered differential-output framework, perfect instantaneous isolation of channel $i$ by a static linear residual is structurally impossible at that state. This is a structural limitation of the system~\eqref{e:sys} at that state.

\section{Neural Observer with Mirror Descent Last Layer Adaptation}\label{s:observer}
The geometric analysis in Section~\ref{s:subspaces} and Theorem~\ref{th:isolability} shows that if the principal angles between a fault signature subspace and its residue subspace are strictly positive, there exist linear maps that generate residuals sensitive to a single channel and decoupled from all others in the output differential space \eqref{e:diffMap}. In this section, we translate this static, differential-geometric picture into a dynamic, implementable scheme. We design a neural network-based observer based on mirror descent last layer adaptation laws for neural networks so that, when the principal angle conditions hold, the learned fault estimates align with the isolable directions suggested by the geometric conditions. The latter is achieved without explicitly computing the output differential map~\eqref{e:diffMap} in real-time.

\subsection{Neural Network-Based Observer}
We consider the following Luenberger-like observer with additive fault estimates:
\begin{align}\label{e:observer}
	\dot{\hat x} &= f(\hat x)+ g(\hat x) u	+ g(\hat x) \hat f_{\mathrm{a}}(\hat x,u) + L(y - \hat y), \nonumber\\
	\hat y &= h(\hat x)	+ \hat f_{\mathrm{s}}(\hat x,u),
\end{align}
where $\hat x\in\mathbb{R}^n$ is the estimated state, $L$ is the observer gain, and $g(\hat{x}) = [g_1(\hat{x})\ \cdots\ g_m(\hat{x})]$. Moreover, $\hat f_\mathrm{a}(\hat x, u)= W_\mathrm{a}^\top \phi_\mathrm{a}(\hat x, u) \in \mathbb{R}^{m}$ and $\hat f_\mathrm{s}(\hat x, u)= W_\mathrm{s}^\top \phi_\mathrm{s}(\hat x,u) \in \mathbb{R}^{p}$ are neural network-based actuator and sensor fault estimators, respectively, where $W_\mathrm{a}\in\mathbb{R}^{n_\mathrm{a}\times m}$ and $W_\mathrm{s}\in\mathbb{R}^{n_\mathrm{s}\times p}$ are the last layer matrices of the neural networks. Furthermore, $\phi_\mathrm{a}:\mathbb{R}^n\times \mathbb{R}^m\to\mathbb{R}^{n_\mathrm{a}}$ and $\phi_\mathrm{s}:\mathbb{R}^n\times \mathbb{R}^m\to\mathbb{R}^{n_\mathrm{s}}$ are smooth feature maps that are the outputs of deep hidden neural network layers trained offline.

\begin{assum}\label{assum:NN_approx}
	There exist unknown bounded ideal weights $W_\mathrm{a}^* \in\mathbb{R}^{n_\mathrm{a}\times m}$ and $W_\mathrm{s}^* \in\mathbb{R}^{n_\mathrm{s}\times p}$ such that 
	\begin{align*}
		f_\mathrm{a}(x,u)&= W_\mathrm{a}^{* \top}\phi_a(x,u) + \varepsilon_\mathrm{a}(t), \\
		f_\mathrm{s}(x,u)&= W_\mathrm{s}^{* \top}\phi_\mathrm{s}(x,u) + \varepsilon_\mathrm{s}(t),
	\end{align*}
	where $f_\mathrm{a}(x,u)=[f_{\mathrm{a},1}\ \cdots f_{\mathrm{a},m}]^\top$, $f_\mathrm{s}(x,u)=\sum_{j=1}^{q_\mathrm{s}} e_j f_{\mathrm{s},j}$, and $\|\varepsilon_a(t)\|\le \bar\varepsilon_a$ and $\|\varepsilon_s(t)\|\le \bar\varepsilon_s$ are bounded approximation errors. Furthermore, we assume the feature maps are Lipschitz:  $\|\phi_\mathrm{a}(x, u)-\phi_\mathrm{a}(\hat x, u)\| \le L_{\phi_\mathrm{a}}^\mathrm{e}\|x-\hat x\|$, $\|\phi_\mathrm{s}(x, u)-\phi_\mathrm{s}(\hat x, u)\| \le L_{\phi_\mathrm{s}}^\mathrm{e}\|x-\hat x\|$, where $L_{\phi_\mathrm{a}}^\mathrm{e}>0$ and $L_{\phi_\mathrm{s}}^\mathrm{e}>0$.
\end{assum}

\begin{assum}\label{assum:contraction}
	Consider \eqref{e:sys} under fault-free conditions and let $e=x-\hat x$. There exists a smooth metric $M(x)=~M(x)^\top\succ 0$ and a constant $\lambda>0$ such that the nominal error dynamics $\dot e = A_\mathrm{L}(x)e$ satisfy the inequality
	\begin{equation*}
		\dot M + A_\mathrm{L}^\top M + M A_\mathrm{L} \preceq -2\lambda M,
	\end{equation*}
	where 
	\begin{equation*}
		A_\mathrm{L}(x) = \frac{\partial f}{\partial x}(x) + \sum_{k=1}^{m} \frac{\partial g_k}{\partial x}(x) u_k - L \frac{\partial h}{\partial x}(x).
	\end{equation*}
    Moreover, there exist scalars $\underline m,\bar m>0$ such that $\underline m I \preceq M(x) \preceq \bar m I$
\end{assum}


The existence of ideal weights in Assumption~\ref{assum:NN_approx} is a standard consideration in neural network–based approximation and adaptive control frameworks, and is commonly used to derive boundedness guarantees. Moreover, Assumption~\ref{assum:contraction} is common in contraction-based observer design methods and can be enforced, for example, via neural network-based parameterization of $M(x)$.

\subsection{Bregman Divergence and Potential Function}
We now define the mirror descent adaptation laws for the last layer matrices. Let $\psi:\mathbb{R}^d\to\mathbb{R}$ be a strictly convex, differentiable potential function (i.e., the mirror map). The associated Bregman divergence is defined as
\begin{equation*}
	D_\psi(\omega\,\|\,\omega^*) := \psi(\omega) - \psi(\omega^*)- \langle\nabla\psi(\omega^*),\omega-\omega^*\rangle,
\end{equation*}
where $\omega, \omega^*\in\mathbb{R}^d$. This divergence measures distance with respect to the geometry induced by $\psi$ (\cite{boffi2021implicit}). Moreover, the Bregman divergence is non-negative and becomes zero if and only if $\omega=\omega^*$. For a time-varying parameter $\omega(t)\in\mathbb{R}^d$ and a constant reference $\omega^*$, the derivative of the Bregman divergence is
\begin{equation}\label{e:BregDot}
	\frac{d}{dt} D_\psi \big(\omega^*\,\|\,\omega(t)\big) = (\omega(t)-\omega^*)^\top \nabla^2 \psi\big(\omega(t)\big)\,\dot \omega(t),
\end{equation}
where $\nabla^2\psi(\cdot) \succ 0$ is the Hessian of the potential function.

Different choices of $\psi$ encode different parameter-space geometries and can be informed by the principal angle study of Section~\ref{s:subspaces}. Thus, to encode both isolable directions and channel-level sparsity consistent with the fault isolation, we adopt the elastic net mirror map. 

\begin{defn}[Elastic Net Mirror Map]\label{def:elastic-net}
	For a weight matrix $W=[w_{(:,1)}\ \cdots\ w_{(:,c)}]\in\mathbb{R}^{n_\mathrm{e}\times c}$, we have
	\begin{equation}\label{e:elastic-net}
		\psi_{\mathrm{EN}}(W) := \frac{\beta}{2}\,\langle W,\,\Xi W\rangle + \alpha \sum_{j=1}^c \sqrt{\|w_{(:,j)}\|_2^2 + \varepsilon},
	\end{equation}
	where $\beta>0$, $\alpha>0$, $\varepsilon>0$, and $\Xi=\mathrm{diag}(\xi_1,\dots,\xi_c)\succ 0$ is a diagonal scaling matrix.
\end{defn}

Intuitively, the elastic-net mirror map combines two complementary effects: the quadratic term, weighted by $\Xi$, induces a Riemannian-like geometry in weight space in which channels associated with smaller minimal principal angles (i.e., geometrically harder to isolate) can be assigned larger curvature and hence smaller effective step size, while the smoothed $\ell_1$-type term promotes channel-wise sparsity, so that adaptation focuses on the few channels actually excited by the current fault. The map in~\eqref{e:elastic-net} satisfies the following properties:
\begin{enumerate}[(a)]
	\item The function $\psi_{\mathrm{EN}}$ is strongly convex and there exists $\lambda_\psi>0$ such that
	\begin{equation*}
		D_{\psi_{\mathrm{EN}}}(W\|W^*) \ge \frac{\lambda_\psi}{2}\,\|W-W^*\|_F^2.
	\end{equation*}

	\item Its Hessian is block-diagonal across columns, and there exists $\underline\lambda>0$ such that
	\begin{equation*}
	\nabla^2\psi_{\mathrm{EN}}(W)\succeq \beta\,\Xi,\qquad \big\|\big(\nabla^2\psi_{\mathrm{EN}}(W)\big)^{-1}\big\|\le \frac{1}{\beta\,\underline\lambda},
	\end{equation*}
	where $\underline\lambda:=\min_j\xi_j$.
	\item Column-separability implies that the variable
	$\Omega_j=\nabla_{w_{(:,j)}}\psi_{\mathrm{EN}}(W)$ depends only on $w_{(:,j)}$, which yields channel-wise decoupling.
\end{enumerate}

The quadratic term in~\eqref{e:elastic-net} defines a Euclidean geometry in weight space where larger $\xi_j$ emphasizes the $j$-th channel by shrinking the corresponding block of the inverse Hessian. Also, the second term in~\eqref{e:elastic-net} is a smooth proxy for an $L_1$ penalty. It promotes channel-level sparsity by discouraging simultaneous activation of many columns.

\subsection{Mirror Descent Adaptive Laws}
Let us define $r:= y - \hat y$. Furthermore, the variable $z:= g(\hat x)^\top M(\hat x) L r$ captures the sensitivity of the contraction energy $e^\top M e$ to perturbations in the actuator fault estimates. We consider the following loss function:
\begin{equation}\label{e:loss}
	\mathcal{L}(\hat x,W_\mathrm{a},W_\mathrm{s})= \frac{1}{2} \|r\|^2,
\end{equation}
which drives the observer to reduce the state estimation error. For the sensor weight $W_\mathrm{s}$, the gradient of~\eqref{e:loss} is
\begin{equation*}
	\nabla_{w_{\mathrm{s},j}}\mathcal{L} = -\phi_\mathrm{s}(\hat x) r_j,
\end{equation*}
where $r_j$ is the $j$-th element of $r$, for $j=1,\dots,p$. Moreover, for the actuator weight $W_a$, the gradient of \eqref{e:loss} can be approximated by $\nabla_{w_{\mathrm{a},i}}\mathcal{L}= \phi_\mathrm{a}(\hat x) z_i$, where $z_i$ is the $i$-th entry of $z$, for $i=1,\dots,m$.

Considering these gradients, the mirror descent-based adaptation laws associated with~\eqref{e:loss} are
\begin{align}\label{e:adaptiveLaw_matrix}
		\dot W_\mathrm{a} &= - \Gamma_\mathrm{a} \big(\nabla^2\psi_\mathrm{a}(W_\mathrm{a})\big)^{-1} \phi_\mathrm{a}(\hat x) z^\top - \sigma_\mathrm{a} W_\mathrm{a}, \nonumber\\
		\dot W_\mathrm{s} &= \Gamma_\mathrm{s} \big(\nabla^2\psi_\mathrm{s}(W_\mathrm{s})\big)^{-1} \phi_\mathrm{s}(\hat x) r^\top - \sigma_\mathrm{s} W_\mathrm{s},
\end{align}
where $\Gamma_\mathrm{a},\Gamma_\mathrm{s}\succ 0$ are diagonal gain matrices and $\sigma_\mathrm{a}, \sigma_\mathrm{s}> 0$ are scalars. Per-channel adaptive laws can be expressed by
\begin{align}\label{e:adaptiveLaw_channel}
		\dot w_{\mathrm{a},i} &= -\Gamma_{\mathrm{a},i}	K_{\mathrm{a},i}(W_\mathrm{a})^{-1}\phi_\mathrm{a}(\hat x) z_i - \sigma_\mathrm{a} w_{\mathrm{a},i}, \nonumber\\
		\dot w_{\mathrm{s},j} &= \Gamma_{\mathrm{s},j} K_{\mathrm{s},j}(W_\mathrm{s})^{-1}\phi_\mathrm{s}(\hat x)r_j - \sigma_\mathrm{s} w_{\mathrm{s},j},
\end{align}
where $K_{\mathrm{a},i}$, $K_{\mathrm{s},j}$, $\Gamma_{\mathrm{a},i}$, and $\Gamma_{\mathrm{s},j}$ are the $i$-th and $j$-th diagonal blocks of $\nabla^2\psi_a(W_a)$, $\nabla^2\psi_s(W_s)$, $\Gamma_\mathrm{a}$, and $\Gamma_\mathrm{s}$, respectively. 

It is worth noting that each Hessian block satisfies $\|K_{\mathrm{a},i}^{-1}(W_a)\|\le\frac{1}{\beta\,\xi_{\mathrm{a},i}}$ and $\|K_{\mathrm{s},j}^{-1}(W_s)\|\le\frac{1}{\beta\,\xi_{\mathrm{s},j}}$. Hence, in order to design elements of $\Xi$, channels whose fault signature subspaces are nearly aligned with the residue subspaces (i.e., small minimal principal angle) can be assigned larger weights $\xi_i$, which increases the local curvature and thus reduces the mirror descent step size in those ambiguous directions. On the other hand, channels with larger principal angles (i.e., better geometric isolability) can be assigned smaller $\xi_i$ so that the mirror descent step is larger along geometrically informative directions.

\begin{rem}
		All fault signature and residue subspaces introduced in Definitions~\ref{def:subspaces} and \ref{def:residue} are implicitly time-varying, since they depend on the Jacobians of the dynamics and output map evaluated along the fault-free trajectory $x^*(t)$. Consequently, the minimal principal angle $\theta_{\min}(\mathcal{F}_{\mathrm{a},i}^{[0:R]}(t),\Pi_i^{[0:R]}(t))$ is also time-varying, and isolability in this context should be interpreted as an instantaneous structural property, where a channel may be isolable at some time instants and non-isolable at others (as reflected Theorems~\ref{th:isolability} and~\ref{th:non-isolability}). The observer \eqref{e:observer} and the mirror descent adaptation laws \eqref{e:adaptiveLaw_matrix} do not require real-time computation of these subspaces or angles. Nevertheless, the geometric information extracted from this analysis can be	used offline to inform the design of the mirror descent geometry. For each actuator or sensor fault channel, one may sample representative fault-free trajectories over the operational region and evaluate the instantaneous minimal principal angle $\theta_{\min}(\mathcal{F}_{\mathrm{a},i}^{[0:R]}(x^*(t)), \Pi_i^{[0:R]}(x^*(t)))$. From these samples, a conservative lower bound or statistically robust estimate of the minimal angle for each channel, denoted by $\underline{\theta}_{\mathrm{a},i}$ and $\underline{\theta}_{\mathrm{s},j}$, can be obtained. These offline estimates, which determine how well each channel is geometrically separable in typical or worst-case operating conditions, are then used to select the curvature parameters in the elastic net mirror map
		(i.e., the diagonal entries of $\Xi$). Hence, the principal angle geometry is incorporated once into the mirror map, while the online observer operates without explicit subspace or angle computation.
	\end{rem}
    
\subsection{Exponential Stability and Robustness Analysis}\label{s:stability}
We now analyze the stability of the neural network-based observer \eqref{e:observer} with mirror descent last layer adaptation, and show that both the state estimation error and the parameter errors are UUB.

Using \eqref{e:sys} and \eqref{e:observer}, the estimation error dynamics can be obtained as
\begin{align}\label{e:error_dyn}
	\dot e	=&\ f(x) - f(\hat x) + \big(g(x)-g(\hat x)\big) u + g(x) f_\mathrm{a}(x,u) \nonumber\\
	& - g(\hat x)\hat f_\mathrm{a}(\hat x,u) - L\Big(h(x)-h(\hat x)\Big) \nonumber\\
	& - L\Big(f_\mathrm{s}(x,u)-\hat f_\mathrm{s}(\hat x,u)\Big) + d(t).
\end{align}
Let $\Delta_\mathrm{nom}(x,\hat x,u)=\ f(x)-f(\hat x)+ \big(g(x)-g(\hat x)\big)u - L\big(h(x)-h(\hat x)\big)- A_\mathrm{L}(\hat x) e$
and the mismatch due to the actuator faults as $\Delta_{\mathrm{fa}}(x,\hat x,u)= \big(g(x)-g(\hat x)\big) f_\mathrm{a}(x,u)$. Consequently, \eqref{e:error_dyn} can be rewritten as
\begin{align}\label{e:edot_split}
	\dot e =&\ A_\mathrm{L}(\hat x)e + \Delta_\mathrm{nom}(x,\hat x,u) + g(\hat x)\big(f_\mathrm{a}(x,u) - \hat f_\mathrm{a}(\hat x,u)\big) \nonumber\\
	&\ + \Delta_{\mathrm{fa}}(x,\hat x,u) - L\big(f_\mathrm{s}(x,u)-\hat f_\mathrm{s}(\hat x,u)\big)
	+ d(t).
\end{align}

As per Assumption \ref{assum:NN_approx}, we decompose
\begin{align}\label{e:fa_diff}
	f_\mathrm{a}(x,u) - \hat f_\mathrm{a}(\hat x,u)
	=&\ \tilde W_\mathrm{a}^\top \phi_\mathrm{a}(\hat x,u) + W_\mathrm{a}^{*\top}\big(\phi_\mathrm{a}(x,u) \nonumber\\
	& - \phi_\mathrm{a}(\hat x,u)\big) + \varepsilon_\mathrm{a}(t), 
\end{align}
and similarly
\begin{align}\label{e:fs_diff}
	f_\mathrm{s}(x,u) - \hat f_\mathrm{s}(\hat x,u)
	=&\ \tilde W_\mathrm{s}^\top \phi_\mathrm{s}(\hat x,u) + W_\mathrm{s}^{*\top}\big(\phi_\mathrm{s}(x,u) \nonumber\\
	& - \phi_\mathrm{s}(\hat x,u)\big) + \varepsilon_\mathrm{s}(t),
\end{align}
where $\tilde W_\mathrm{a}= W_\mathrm{a}^*-W_\mathrm{a}$ and $\tilde W_\mathrm{s}= W_\mathrm{s}^*-W_\mathrm{s}$.

Substituting~\eqref{e:fa_diff} and \eqref{e:fs_diff} into~\eqref{e:edot_split}, one obtains
\begin{align}\label{e:edot_full}
	\dot e =&\ A_\mathrm{L}(\hat x)e
	+ g(\hat x)\tilde W_\mathrm{a}^\top \phi_\mathrm{a}(\hat x,u)
	- L\tilde W_\mathrm{s}^\top \phi_\mathrm{s}(\hat x,u) \nonumber\\
	&\ + \Delta_\mathrm{nom}(x,\hat x,u)
	+ \Delta_{\mathrm{fa}}(x,\hat x,u) \nonumber\\
	&\ + g(\hat x)W_\mathrm{a}^{*\top}\big(\phi_\mathrm{a}(x,u) - \phi_\mathrm{a}(\hat x,u)\big) \nonumber\\
	&\ - L W_\mathrm{s}^{*\top}\big(\phi_\mathrm{s}(x,u) - \phi_\mathrm{s}(\hat x,u)\big) \nonumber\\
	&\ + g(\hat x)\varepsilon_\mathrm{a}(t)
	- L\varepsilon_\mathrm{s}(t)
	+ d(t).
\end{align}

Let $D_{\psi_\mathrm{a}}$ and $D_{\psi_\mathrm{s}}$ denote the Bregman divergences associated with the strictly convex mirror maps $\psi_\mathrm{a}$ and $\psi_\mathrm{s}$ (i.e., the elastic net mirror map in Definition~\ref{def:elastic-net}). We consider the following Lyapunov candidate:
\begin{align}\label{e:V}
	V(e,\tilde W_\mathrm{a},\tilde W_\mathrm{s})=& \frac{1}{2} e^\top M(\hat x) e + \sum_{i=1}^{m} \frac{1}{\gamma_{\mathrm{a},i}} D_{\psi_\mathrm{a}}\big(w_{\mathrm{a},i}^*\,\|\,w_{\mathrm{a},i}\big) \nonumber\\
	& + \sum_{j=1}^{p} \frac{1}{\gamma_{\mathrm{s},j}} D_{\psi_\mathrm{s}}\big(w_{\mathrm{s},j}^*\,\|\,w_{\mathrm{s},j}\big),
\end{align}
where $w_{\mathrm{a},i}$, $w_{\mathrm{s},j}$, $\gamma_{\mathrm{a},i}$, and $\gamma_{\mathrm{s},j}>0$ are the $i$-th and $j$-th columns of $W_\mathrm{a}$ and $W_\mathrm{s}$, and the diagonal entries of $\Gamma_\mathrm{a}$ and $\Gamma_\mathrm{s}$, respectively.

The derivative of \eqref{e:V} along the trajectories of \eqref{e:sys} and \eqref{e:observer} yields
\begin{align}\label{e:Vdot}
	\dot V &= \frac{d}{dt}\big(\frac{1}{2}e^\top M(\hat x)e\big)
	+ \sum_{i=1}^{m} \frac{1}{\gamma_{\mathrm{a},i}} \frac{d}{dt} D_{\psi_\mathrm{a}}\big(w_{\mathrm{a},i}^*\,\|\,w_{\mathrm{a},i}\big) \nonumber\\
	&\quad + \sum_{j=1}^{p} \frac{1}{\gamma_{\mathrm{s},j}} \frac{d}{dt} D_{\psi_\mathrm{s}}\big(w_{\mathrm{s},j}^*\,\|\,w_{\mathrm{s},j}\big).
\end{align}
Considering \eqref{e:edot_full}, for the first term of \eqref{e:Vdot}, we have
\begin{align}\label{e:Vdot_state}
	&\frac{d}{dt} \big( \frac{1}{2} e^\top M e\big)
	\!=\! \tfrac{1}{2}\,e^\top\!\big(\dot M \!+\! A_\mathrm{L}^\top M\!+\!M A_\mathrm{L}\big)e + e^\top M g(\hat x) \\
	&\!\times\!\tilde W_\mathrm{a}^\top \phi_\mathrm{a}(\hat x,u)\!-\! e^\top M L \tilde W_\mathrm{s}^\top \phi_\mathrm{s}(\hat x,u) 
	\!+\! e^\top M \Delta_\mathrm{rem}(x,\hat x,u,t), \nonumber
\end{align}
where we have defined the aggregate remainder
\begin{align*}
	\Delta_\mathrm{rem} =& \Delta_\mathrm{nom}(x,\hat x,u)
	+ \Delta_{\mathrm{fa}}(x,\hat x,u) + g(\hat x)W_\mathrm{a}^{*\top} \\
	& \times \big(\phi_\mathrm{a}(x,u) - \phi_\mathrm{a}(\hat x,u)\big) - L W_\mathrm{s}^{*\top}\big(\phi_\mathrm{s}(x,u) \\
	& - \phi_\mathrm{s}(\hat x,u)\big) + g(\hat x)\varepsilon_\mathrm{a}(t) - L\varepsilon_\mathrm{s}(t) + d(t).
\end{align*}

From \eqref{e:BregDot} and the adaptive laws \eqref{e:adaptiveLaw_channel}, we obtain
\begin{align*}
	\frac{1}{\gamma_{\mathrm{a},i}}\frac{d}{dt}D_{\psi_\mathrm{a}}\big(\!w_{\mathrm{a},i}^*\,\|\,w_{\mathrm{a},i}\!\big)
	\!=&\! -\frac{1}{\gamma_{\mathrm{a},i}}\tilde w_{\mathrm{a},i}^\top K_{\mathrm{a},i}\dot w_{\mathrm{a},i} \\
	\!=&\! \tilde w_{\mathrm{a},i}^\top \phi_\mathrm{a}(\hat x,u) z_i + \frac{\sigma_\mathrm{a}}{\gamma_{\mathrm{a},i}}\tilde w_{\mathrm{a},i}^\top K_{\mathrm{a},i}w_{\mathrm{a},i},
\end{align*}
where $z_i$ is the $i$-th element of $z$. Summing over $i$ and stacking columns yields
\begin{equation}\label{e:Vdot_param_a}
	\sum_{i=1}^{m} \frac{1}{\gamma_{\mathrm{a},i}}
	\frac{d}{dt}D_{\psi_\mathrm{a}}\big(w_{\mathrm{a},i}^*\,\|\,w_{\mathrm{a},i}\big)
	= z^\top \tilde W_\mathrm{a}^\top \phi_\mathrm{a}(\hat x,u) + \mathcal{R}_\mathrm{a},
\end{equation}
where $\mathcal{R}_\mathrm{a} = \sigma_\mathrm{a}\sum_{i}\tfrac{1}{\gamma_{\mathrm{a},i}}\tilde w_{\mathrm{a},i}^\top K_{\mathrm{a},i}w_{\mathrm{a},i}$. Analogously, for the sensor channels, one has
\begin{equation}\label{e:Vdot_param_s_correct}
	\sum_{j=1}^{p} \frac{1}{\gamma_{\mathrm{s},j}}
	\frac{d}{dt}D_{\psi_\mathrm{s}}\big(w_{\mathrm{s},j}^*\,\|\,w_{\mathrm{s},j}\big)
	= - r^\top \tilde W_\mathrm{s}^\top \phi_\mathrm{s}(\hat x,u) + \mathcal{R}_\mathrm{s},
\end{equation}
where $\mathcal{R}_\mathrm{s} = \sigma_\mathrm{s}\sum_{j}\tfrac{1}{\gamma_{\mathrm{s},j}}\tilde w_{\mathrm{s},j}^\top K_{\mathrm{s},j}w_{\mathrm{s},j}$.

By linearizing $r$, $z$, $h$, and the observer dynamics around $(x,\hat x)$, it can be shown that $e^\top M g(\hat x)\tilde W_\mathrm{a}^\top \phi_\mathrm{a}(\hat x,u)= -z^\top \tilde W_\mathrm{a}^\top \phi_\mathrm{a}(\hat x,u)+ \delta_{\mathrm{a}}(e)$ and $\, -e^\top M L\tilde W_\mathrm{s}^\top \phi_\mathrm{s}(\hat x,u)\, = \, r^\top \tilde W_\mathrm{s}^\top \phi_\mathrm{s}(\hat x,u)+ \delta_{\mathrm{s}}(e)$, where $\delta_{\mathrm{a}}(e)$ and $\delta_{\mathrm{s}}(e)$ are higher-order terms arising from the linearization and Lipschitz properties, satisfying $\|\delta_{\mathrm{a}}(e)\| \le c_{\mathrm{a}} \|e\|^2$ and $\|\delta_{\mathrm{s}}(e)\| \le c_{\mathrm{s}} \|e\|^2$, for some positive constants $c_{\mathrm{a}}$ and $c_{\mathrm{s}}$.

Using $\tilde w_{\mathrm{a},i}^\top K_{\mathrm{a},i}w_{\mathrm{a},i} = -\tilde w_{\mathrm{a},i}^\top K_{\mathrm{a},i}\tilde w_{\mathrm{a},i} + \tilde w_{\mathrm{a},i}^\top K_{\mathrm{a},i}w^*_{\mathrm{a},i}$, Young's inequality, and property~(b) of $\psi_\mathrm{EN}$, one has
\begin{equation*}
	\mathcal{R}_\mathrm{a} \le -\kappa_\mathrm{a}\|\tilde W_\mathrm{a}\|_F^2 + c_\mathrm{a}^*, \quad
	\mathcal{R}_\mathrm{s} \le -\kappa_\mathrm{s}\|\tilde W_\mathrm{s}\|_F^2 + c_\mathrm{s}^*,
\end{equation*}
where $\kappa_\mathrm{a},\kappa_\mathrm{s}>0$ proportional to $\sigma_\mathrm{a},\sigma_\mathrm{s}$ and $c_\mathrm{a}^*,c_\mathrm{s}^*>0$ depending only on the bounded ideal weights $W_\mathrm{a}^*,W_\mathrm{s}^*$. Consequently, one obtains
\begin{align}\label{e:Vdot_bounds}
	\dot V \le -\lambda\,e^\top M e &- \kappa_\mathrm{a}\|\tilde W_\mathrm{a}\|_F^2 - \kappa_\mathrm{s}\|\tilde W_\mathrm{s}\|_F^2 + \delta_{\mathrm{a}}(e) \nonumber \\
	&+ \delta_{\mathrm{s}}(e) + e^\top M \Delta_\mathrm{rem}(x,\hat x,u,t) + c^*,
\end{align}
where $c^*:=c_\mathrm{a}^*+c_\mathrm{s}^*$. Using property~(a) of $\psi_\mathrm{EN}$, one has $V \le \tfrac{\bar m}{2}\|e\|^2 + \tfrac{L_\psi}{2\gamma_{\min}}\big(\|\tilde W_\mathrm{a}\|_F^2+\|\tilde W_\mathrm{s}\|_F^2\big)$ for some $L_\psi>0$, so~\eqref{e:Vdot_bounds} can be rearranged into
\begin{equation}\label{e:Vdot_V}
	\dot V \le -\alpha_V V + \sigma,
\end{equation}
for some $\alpha_V>0$ and $\sigma>0$, provided $\sigma_\mathrm{a},\sigma_\mathrm{s}$ are large enough that $\kappa_\mathrm{a},\kappa_\mathrm{s}$ dominate the residual $\|\tilde W\|^2$ coefficients arising from $\Delta_\mathrm{rem}$.

\begin{assum}\label{assum:Delta_rem_bound}
	We consider $f$, $g$, and $h$ to be locally Lipschitz on a compact set that contains the trajectories $(x(t),\hat x(t),u(t))$ for all $t\ge 0$.
\end{assum}
As a direct consequence of Assumptions~\ref{assum:NN_approx}, \ref{assum:contraction}, and \ref{assum:Delta_rem_bound}, there exist constants $c_0,c_1>0$, such that, for all $(x,\hat x,u,t)$ in the above sets, one has $\big\|\Delta_\mathrm{rem}(x,\hat x,u,t)\big\| \le c_0 + c_1 \|e\|$.

\begin{thm}[Boundedness of Estimations]\label{thm:UUB_observer}
	Consider the system \eqref{e:sys} under simultaneous actuator and sensor faults, the observer \eqref{e:observer} with mirror descent last layer adaptation, and the Lyapunov candidate \eqref{e:V}. Suppose Assumptions~\ref{assum:NN_approx}, \ref{assum:contraction}, and \ref{assum:Delta_rem_bound} hold and that $\sigma_\mathrm{a},\sigma_\mathrm{s}>0$ are chosen such that $\alpha_V>0$ in~\eqref{e:Vdot_V}. The estimation error $e(t)$ and the weight parameter errors $\tilde W_\mathrm{a}(t)$ and $\tilde W_\mathrm{s}(t)$ are uniformly ultimately bounded. Moreover, $e(t)$ converges exponentially to a residual set of size $\sigma/\alpha_V$.
\end{thm}

\begin{pf}
	From~\eqref{e:Vdot_V} and the assumptions, we obtain
	\begin{equation*}
		V(t) \le V(0)\,e^{-\alpha_V t} + \frac{\sigma}{\alpha_V},
	\end{equation*}
	so $V(t)$ is bounded and converges exponentially to the residual set $\{V\le \sigma/\alpha_V\}$. Since $V \ge \tfrac{\underline m}{2}\|e\|^2$, the state estimation error $e(t)$ is uniformly ultimately bounded with exponential convergence. By property~(a) of $\psi_\mathrm{EN}$, $V \ge \tfrac{\lambda_\psi}{2\gamma_{\max}}\big(\|\tilde W_\mathrm{a}\|_F^2 + \|\tilde W_\mathrm{s}\|_F^2\big)$, so $\tilde W_\mathrm{a}(t)$ and $\tilde W_\mathrm{s}(t)$ are also uniformly ultimately bounded. \qed
\end{pf}

\section{Numerical Case Study}
We consider a $3$-axis attitude control system with loss of effectiveness actuator and sensor faults for a spacecraft with four reaction wheels in a pyramid configuration. If the $i$-th and $j$ actuator and sensor are faulty, $g_i^\mathrm{f}(x)=g_i(x)$, $f_{\mathrm{a},i}=(\eta_{\mathrm{a},i}-1)u_i$, the $j$-th row of $e_j$ is equal to $1$, and $f_{\mathrm{s},j}(t)$ is a time-varying additive sensor fault signal, where $\eta_{\mathrm{a},i}$ is the $i$-th actuator effectiveness. The spacecraft inertia matrix is given by $I=\mathrm{diag}(1,1,0.8)$ ($\mathrm{kg}\cdot \mathrm{m}^2$), each reaction wheel has an inertia of $J_\omega=0.01$ ($\mathrm{kg}\cdot \mathrm{m}^2$), and the maximum deliverable torque per wheel is limited to $0.14$ ($\mathrm{N}\cdot \mathrm{m}$). The parameters in (\cite{lee2017geometric}) are used as the spacecraft and reaction wheels parameters. Gains of a nominal controller $u_\mathrm{nom}$ are tuned as $K_p=\mathrm{diag}(22.5,18,15)$ and $K_d=\mathrm{diag}(12,9,7.5)$. Let $\theta=[\theta_1,\, \theta_2, \, \theta_3]^\top \in \mathbb{R}^3$ ($\mathrm{rad}$) where $\theta_1$, $\theta_2$, and $\theta_3$ are roll, pitch, and yaw angles of the spacecraft, respectively. We define the state estimation errors between the system \eqref{e:sys} and the observer \eqref{e:observer} for roll, pitch, and yaw as $e_{\theta_1}$, $e_{\theta_2}$, and $e_{\theta_3}$, respectively. The trained neural networks $\hat f_\mathrm{a}(\hat x,u)$ and $\hat f_\mathrm{s}(\hat x,u)$ have $4$ layers with $({\theta}, \dot{{\theta}}, u)$ as their input.

We consider the case of actuator faults, where the effectiveness of wheels $2$ and $3$, starting from $t=20$ to $50$ (s) and $t=10$ to $35$ (s), become $\eta_{\mathrm{a},2}=0.25$ and $\eta_{\mathrm{a},3}=0.5$, respectively. As can be seen in Fig.~\ref{fig:estimationError}, our hybrid observer using mirror descent (MD) adaptive laws has a small estimation error for all angles. As a comparison, we have shown the results of a gradient descent (GD) adaptive law with a quadratic loss in Fig.~\ref{fig:estimationError}, which demonstrates similar performance to the MD-based approach. However, in the case of fault parameter estimation, as shown in Fig.~\ref{fig:actFaultEst}, the MD-based approach has a more accurate estimate of actuator faults. As for the case of simultaneous actuator and sensor faults, we considered the previous actuator fault profiles and additive sinusoidal sensor faults $f_{\mathrm{s},1}(t)=0.035\sin(0.6\pi(t-15))$ on the roll channel over $t\in[15,50]$ (s) and $f_{\mathrm{s},2}(t)=-0.030\cos(0.5\pi(t-20))$ on the pitch channel over $t\in[20,55]$ (s), while the yaw sensor remains healthy. As shown in Fig.~\ref{fig:senFaultEst}, the MD-based method tracks the active roll and pitch sensor faults more closely than GD and keeps the healthy yaw-channel estimate near zero, whereas the GD estimate drifts away from zero.

\begin{figure}
	\begin{center}
			\includegraphics[width=8.4cm]{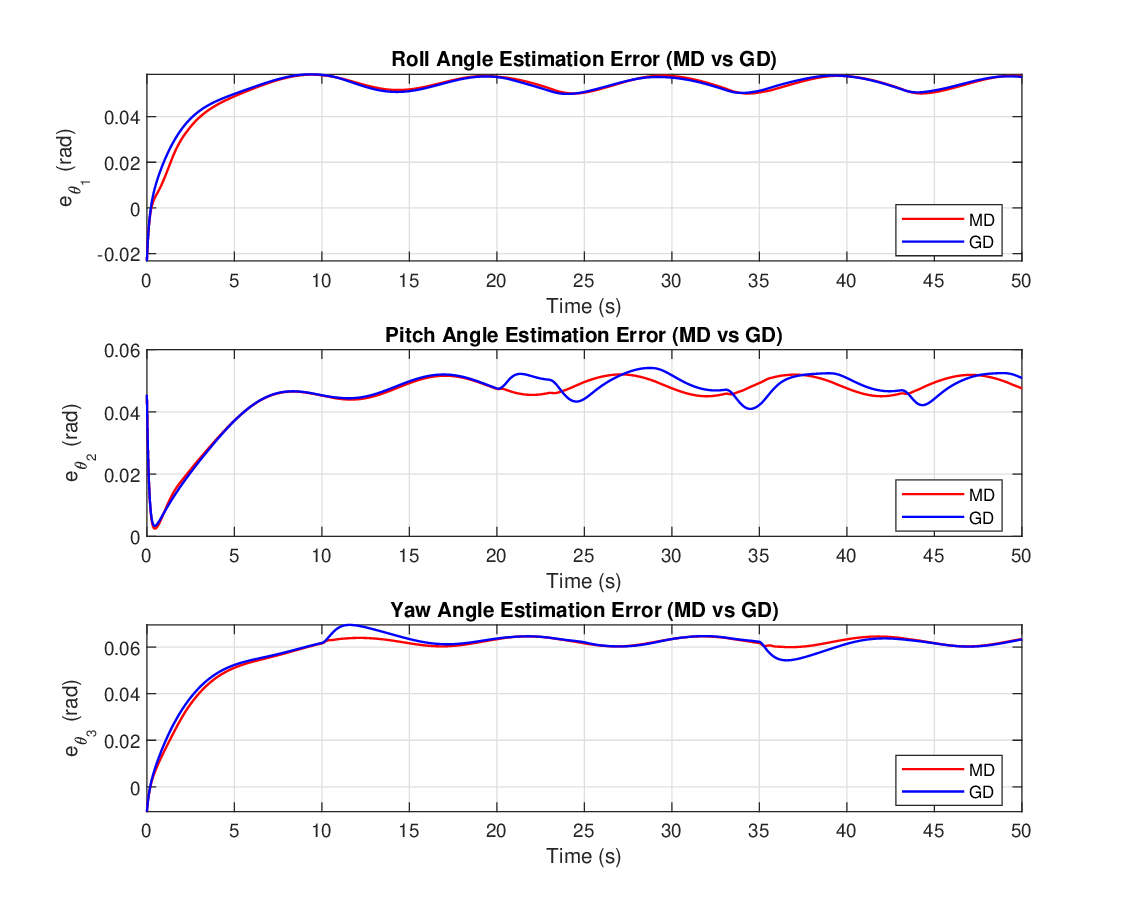}    
			\caption{State estimation error using gradient descent (GD) and mirror descent (MD).} 
			\label{fig:estimationError}
		\end{center}
\end{figure}

\begin{figure}
	\begin{center}
			\includegraphics[width=8.4cm]{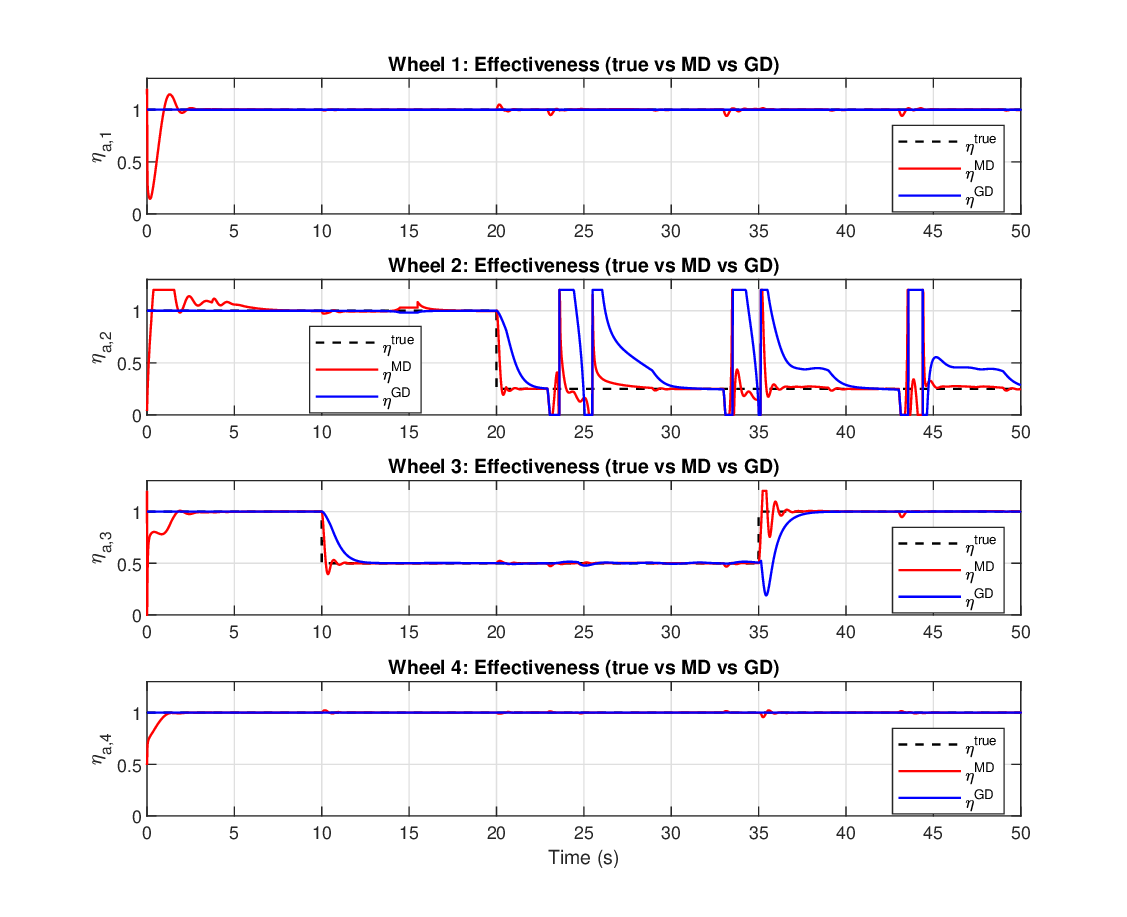}    
			\caption{Actuator loss of effectiveness fault estimation using GD and MD.} 
			\label{fig:actFaultEst}
		\end{center}
\end{figure}

\begin{figure}
	\begin{center}
			\includegraphics[width=7.4cm]{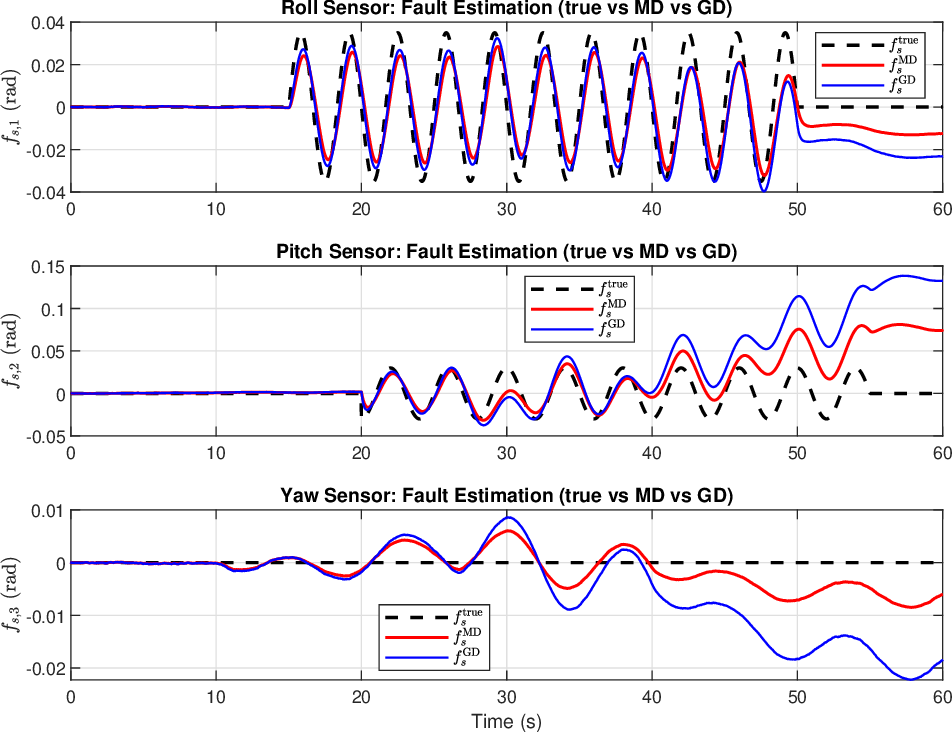}    
			\caption{Sensor fault estimation under simultaneous actuator and sensor faults using GD and MD.} 
			\label{fig:senFaultEst}
		\end{center}
\end{figure}
%
%

\section{Conclusion}
In this paper, we presented a geometric framework for fault detection and identification (FDI) in nonlinear control-affine systems subject to simultaneous actuator and sensor faults. By representing fault effects as signature subspaces in the output differential map, we derived necessary and sufficient isolability conditions in terms of intersections and minimal principal angles between fault and residue subspaces. In order to accurately estimate the system state and fault parameters, we designed a contraction-based Luenberger-like observer augmented with neural network fault estimators whose last layers are adapted online. We utilized a mirror descent-based approach with an elastic net mirror map. Consequently, a Lyapunov analysis combining the contraction metric and Bregman divergences established uniform ultimate boundedness (UUB) of both the state and parameter estimation errors in the presence of bounded disturbances and approximation errors. A numerical case study on a $3$-axis spacecraft attitude control system with loss of effectiveness actuator and sensor faults demonstrated that the proposed method can accurately detect, isolate, and estimate multiple concurrent faults.

                                                   
\end{document}